\documentclass[letterpaper,10 pt,conference]{ieeeconf}  % Comment this line out
\IEEEoverridecommandlockouts                              % This command is only
\usepackage{amsthm} % assumes amsmath package installed
\usepackage{amssymb}  % assumes amsmath package installed

\newtheorem{theorem}{Theorem}
\newtheorem{corollary}{Corollary}
\newtheorem{lemma}{Lemma}
\newtheorem{remark}{Remark}
\newtheorem{definition}{Definition}
\newtheorem{proposition}{Proposition}
\newtheorem{example}{Example}
\newtheorem{assumption}{Assumption}

\usepackage{amsmath,amsfonts,amssymb,color}
\usepackage{epsfig}
\usepackage{algorithm}
\usepackage{algorithmic}
\usepackage{multirow}

\usepackage{array}
\usepackage{multirow}
\usepackage{epstopdf}
\usepackage{tikz}
\usepackage{relsize}
\usetikzlibrary{shapes,snakes,arrows}
\usepackage{cite}
\usepackage{pmat}
\usepackage{epstopdf}
\usepackage{tikz}
\usetikzlibrary{shapes,snakes}
\usepackage{scalerel}
\usepackage[hidelinks]{hyperref}
\hypersetup{
    colorlinks,
    linkcolor={red!},
    citecolor={blue!}
}
\DeclareMathOperator*{\Bigcdots}{\scalerel*{\cdots}{\bigodot}}

\title{\LARGE \bf
Secure Distributed State Estimation of an LTI system over Time-varying Networks and Analog Erasure Channels
}

\author{Aritra~Mitra~and~Shreyas Sundaram% <-this % stops a space
\thanks{The authors are with the School of Electrical and Computer Engineering at Purdue University. Email: {\tt \{mitra14,sundara2\}@purdue.edu}}%
}

\begin{document}
\maketitle
\thispagestyle{empty}
\pagestyle{empty}

%%%%%%%%%%%%%%%%%%%%%%%%%%%%%%%%%%%%%%%%%%%%%%%%%%%%%%%%%%%%%%%%%%%%%%%%%%%%%%%%

\begin{abstract} We study the problem of collaboratively estimating the state of an LTI system monitored by a network of sensors, subject to the following important practical considerations: (i) certain sensors might be arbitrarily compromised by an adversary and (ii) the underlying communication graph governing the flow of information across sensors might be time-varying. We first analyze a scenario involving intermittent communication losses that preserve certain information flow patterns over bounded intervals of time. By equipping the sensors with adequate memory, we show that one can obtain a fully distributed, provably correct state estimation algorithm that accounts for arbitrary adversarial behavior, provided certain conditions are met by the network topology. We then argue that our approach can handle bounded communication delays as well. Next, we explore a case where each communication link stochastically drops packets based on an analog erasure channel model. For this setup, we propose state estimate update and information exchange rules, along with conditions on the network topology and packet drop probabilities, that guarantee mean-square stability despite arbitrary adversarial attacks.
\end{abstract}

\section{Introduction}
Consider a scenario where a group of  sensors deployed over a geographical network seek to cooperatively estimate the state of a dynamical process. This general setup constitutes the standard distributed state estimation problem and finds applications in various domains such as power systems, transportation networks, automated factories, and distributed robotics. As envisioned in \cite{survey1}, to fully leverage the benefits of a distributed sensor network as described above, one must design algorithms and networks that can respond to dynamic environments involving unreliable components. In particular, the growing need for designing secure networked control systems necessitates the design of localized algorithms that operate reliably in the face of adversarial sensor attacks. In this context, unavoidable network-induced issues such as communication losses and delays offer additional degrees of freedom for adversaries to devise carefully crafted attacks, thereby significantly compounding the estimation problem.\footnote{We discuss how adversaries can use communication losses and delays to their advantage later in the paper.} In light of the above discussion, the design of attack-resilient, provably correct distributed state estimation algorithms that account for various types of communication losses and delays will be the subject of our present investigation. 

\textbf{Related Work:} The classical distributed state estimation problem as described above, has been studied extensively over the past decade; however, single-time-scale algorithms that solve such problems under the most general conditions on the system and network have been proposed only recently in \cite{martins3,allerton,mitra16arxiv,wang}. While the problem of detecting and mitigating various forms of data-injection attacks in deterministic \cite{fawzi,pasqualetti,joao2} and stochastic \cite{bai1,bai2} centralized\footnote{Here, by a centralized system, we refer to a system where the measurements of all the sensors are available at a single location.} systems is now well understood, tackling adversarial behavior in the context of distributed state estimation remains largely unexplored. The limited literature that seeks to address this problem either provides no theoretical guarantees \cite{sec1,sec3} or limits the class of admissible attacks \cite{deghat}. In an initial effort to bridge the gap between centralized and distributed secure state estimation, we recently developed a distributed observer that allows each non-compromised sensor to asymptotically recover the entire state dynamics despite \textit{arbitrary} adversarial sensor attacks, under appropriate conditions on the network topology \cite{mitraCDC}. However, our proposed technique did not account for the challenges introduced by communication drop-outs or delays.  This leads us to the contributions of the present paper.

\textbf{Contributions:} Our contributions are twofold. First, in Section \ref{sec:SWLFSE}, we consider a communication loss model where certain information flow patterns are preserved deterministically over bounded intervals of time. For this communication loss model, we show how sensors equipped with memory can process delayed state estimates received from neighbors (some of whom can be potentially adversarial) to asymptotically estimate the state of the system. Our algorithm is inspired by recent work that addresses the resilient consensus problem in asynchronous settings \cite{saldana17}, \cite{dibaji17}. As a byproduct of our analysis, we argue that the proposed  strategy accounts for bounded communication delays as well. We also characterize the convergence rate of our algorithm in terms of the system instability, the upper bound on the delay, and certain properties of the underlying communication graph. The second main result of the paper, presented in Section \ref{sec:analog}, pertains to a network whose communication links are modeled as analog erasure channels that may or may not introduce random delays. We show how a graph metric known as `strong-robustness' (introduced in our prior work \cite{mitraCDC}) can help tolerate higher erasure probabilities, while guaranteeing mean-square stability of the estimation error dynamics. Finally, we emphasize that all our results apply to a sophisticated and worst-case adversarial model (termed Byzantine adversaries) which is typically considered in the literature on resilient distributed algorithms \cite{vaidyacons,rescons,Sundaramopt,su} .
\section{System and Attack Model}
\textbf{Notation:} A directed graph is denoted by $\mathcal{G} =(\mathcal{V},\mathcal{E})$, where $\mathcal{V} =\{1, \cdots, N\}$ is the set of nodes and $\mathcal{E} \subseteq \mathcal{V} \times \mathcal{V} $ represents the edges. An edge from node $j$ to node $i$, denoted by (${j,i}$), implies that node $j$ can transmit information to node $i$. The neighborhood of the $i$-th node is defined as $\mathcal{N}_i \triangleq \{j\,|\,(j,i) \in \mathcal{E} \}.$ The notation $|\mathcal{V}|$ is used to denote the cardinality of a set $\mathcal{V}$. Throughout the rest of this paper, we use the terms `edges' and `communication links/channels' interchangeably. The set of all eigenvalues (modes) of a matrix $\mathbf{A}$ is denoted by $sp(\mathbf{A}) = \{\lambda \in \mathbb{C}\,|\,det(\mathbf{A}-\lambda\mathbf{I}) = 0\}$ and the set of all marginally stable and unstable eigenvalues of $\mathbf{A}$ is denoted by $\Lambda_{U}(\mathbf{A}) = \{\lambda \in sp(\mathbf{A})\,|\, |\lambda| \geq 1 \}$. We use $diag(\mathbf{A}_1,\cdots,\mathbf{A}_2)$ to refer to a block-diagonal matrix with the matrix $\mathbf{A}_i$ as the $i$-th block entry. The notation $\mathbb{Z}_{\geq 0}$ is used to denote the set of all non-negative integers, and for a random variable $\mathbb{X}$, its expected value is denoted by $E[\mathbb{X}]$.

\textbf{System Model:} Consider the linear dynamical system
\begin{equation}
\mathbf{x}[k+1] = \mathbf{Ax}[k],
\label{eqn:plant}
\end{equation}
where $k \in \mathbb{Z}$ is the discrete-time index, $\mathbf{x}[k] \in {\mathbb{R}}^n$ is the state vector and  $\mathbf{A} \in {\mathbb{R}}^{ n \times n} $ is the system matrix. The system is monitored by a network $\mathcal{G}=(\mathcal{V,E})$ consisting of $N$ nodes. The $i$-th node has partial measurement of the state $\mathbf{x}[k]$:
\begin{equation}
\mathbf{y}_{i}[k]=\mathbf{C}_i\mathbf{x}[k],
\label{eqn:Obsmodel}
\end{equation}
where $\mathbf{y}_{i}[k] \in {\mathbb{R}}^{r_i}$ and $\mathbf{C}_i \in {\mathbb{R}}^{r_i \times n}$. We denote $\mathbf{y}[k]={\begin{bmatrix}\mathbf{y}^T_{1}[k] & \cdots & \mathbf{y}^T_{N}[k]\end{bmatrix}}^T$, and $\mathbf{C}={\begin{bmatrix}\mathbf{C}^T_{1} & \cdots & \mathbf{C}^T_{N}\end{bmatrix}}^{T}$. 

In the standard distributed state estimation problem, each node is required to  estimate the state $\mathbf{x}[k]$ using its own measurements and the information received from its neighbors, where such information flow is restricted by the underlying communication graph $\mathcal{G}$. A challenging scenario emerges when one seeks to solve the same problem in the presence of malicious nodes in the network. We now formally describe the adversary model considered throughout the paper.  

\textbf{Adversary Model:} We consider a subset $\mathcal{A} \subset \mathcal{V}$ of the nodes in the network to be adversarial. We assume that the adversarial nodes are completely aware of the network topology (and any variations to such topology due to communication drop-outs), the system dynamics and the algorithm employed by the non-adversarial nodes. In terms of capabilities, an adversarial node can leverage the aforementioned information to arbitrarily deviate from the rules of any prescribed algorithm, while colluding with other adversaries in the process. Furthermore, following the Byzantine fault model\cite{Byz}, adversaries are allowed to send differing state estimates to different neighbors at the same instant of time. This assumption of omniscient adversarial behavior is standard in the literature on resilient and secure distributed algorithms \cite{vaidyacons,rescons,Sundaramopt,su}, and allows us to provide guarantees against ``worst-case" adversarial behavior. In terms of their density in the network, we assume that there are at most $f$ adversarial nodes in the neighborhood of any non-adversarial node; this property will be referred to as the `$f$-local' property of the adversarial set. Summarily, the adversary model described thus far will be called an $f$-local Byzantine adversary model.  The non-adversarial nodes will be referred to as regular nodes and be represented by the set $\mathcal{R}=\mathcal{V}\setminus\mathcal{A}$. Finally, note that the number and identities of the adversarial nodes are not known to the regular nodes.

\textbf{Objective:} Given the LTI system \eqref{eqn:plant}, the measurement model \eqref{eqn:Obsmodel}, a communication graph $\mathcal{G}$, and the $f$-local Byzantine adversary model described above, our  objective in this paper is to design state estimate update and information exchange rules that guarantee asymptotic convergence (in a deterministic or stochastic sense) of the estimates of the regular nodes to the true state of the plant, under different types of communication loss models.

\section{Preliminaries}
Before developing our estimation strategy, we first establish certain terminology, notation, and key ideas in this section. To begin with, the underlying communication graph $\mathcal{G}$ that dictates the flow of information among nodes in the absence of any communication losses will be referred to as the \textit{baseline communication graph}. The loss of communication between nodes is modeled by a time-varying graph $\mathcal{G}[k]=(\mathcal{V},\mathcal{E}[k])$, where $\mathcal{E}[k] \subseteq \mathcal{E}$. Regarding system \eqref{eqn:plant}, we make the following assumption for clarity of exposition.\footnote{The results presented in this paper can however be extended to system matrices with arbitrary spectrum via a more detailed technical analysis.}

\begin{assumption} The system matrix $\mathbf{A}$ has real, distinct eigenvalues.
\end{assumption}

Based on the above assumption, one can perform a coordinate transformation $\mathbf{z}[k]\triangleq \mathbf{Vx}[k]$ on \eqref{eqn:plant}, with an appropriate non-singular matrix $\mathbf{V}$, to obtain\footnote{As this only relies on the knowledge of the system matrix $\mathbf{A}$ (which is assumed to be known by all the nodes), all of the nodes can do this in a distributed manner.}
\begin{equation}
\begin{split}
\mathbf{z}[k+1] &= \mathbf{Mz}[k] = diag(\lambda_1, \cdots, \lambda_n)\mathbf{z}[k],\\
\mathbf{y}_i[k] &= \bar{\mathbf{C}}_i\mathbf{z}[k], \quad \forall i \in \{1, \cdots, N\} \enspace 
\end{split}
\label{eqn:plant_tr}
\end{equation}
where $sp(\mathbf{A})=diag(\lambda_1, \cdots, \lambda_n)$, $\mathbf{M}=\mathbf{VA}\mathbf{V}^{-1}$ and $\bar{\mathbf{C}}_i=\mathbf{C}_i\mathbf{V}^{-1}$. Commensurate with this decomposition, the $j$-th component  of the state vector $\mathbf{z}[k]$ will be denoted by $z^{(j)}[k]$, and will be referred to as the component corresponding to the eigenvalue $\lambda_j$. Since recovering $\mathbf{z}[k]$ is equivalent to recovering $\mathbf{x}[k]$, we focus on estimating $\mathbf{z}[k]$. To this end, we will use the following definition of \textit{source nodes}. 

\begin{definition}
(\textbf{Source nodes})
For each $\lambda_j \in \Lambda_{U}(\mathbf{A})$, the set of nodes that can detect $\lambda_j$ is denoted by $\mathcal{S}_j$, and called the set of \textit{source nodes} for $\lambda_j$. 
\end{definition}

Let $\Omega_{U}(\mathbf{A}) \subseteq \Lambda_{U}(\mathbf{A})$ contain the set of eigenvalues of $\mathbf{A}$ for which $\mathcal{V}\setminus{\mathcal{S}_j}$ is non-empty. Then, for each $\lambda_j \in \Omega_{U}(\mathbf{A})$, our strategy requires the source nodes $\mathcal{S}_j$ to maintain local\footnote{Here, by `local' we imply that such observers can be constructed and run without any information from neighbors.} Luenberger observers for estimating $z^{(j)}[k]$, while the non-source nodes rely on a secure consensus protocol for the same. For any node $i$, let the set of eigenvalues it can detect be denoted by $\mathcal{O}_i$, and let $\mathcal{UO}_i=sp(\mathbf{A})\setminus\mathcal{O}_i$. Then, the following result from \cite{mitraCDC} states that node $i$ can estimate the locally detectable portion of $\mathbf{z}[k]$, referred to as $\mathbf{z}_{\mathcal{O}_i}[k]$, \textit{without} interacting with its neighbors.

\begin{lemma}
Suppose Assumption 1 holds. Then, for each regular node $i\in\mathcal{R}$ and each $\lambda_j \in \mathcal{O}_{i}$, a local Luenberger observer can be constructed that ensures that $\lim_{k\to\infty}|\hat{z}^{(j)}_i[k]-z^{(j)}[k]|=0$, where $\hat{z}^{(j)}_i[k]$ denotes the estimate of $z^{(j)}[k]$ maintained by node $i$.
\label{lemma:luen}
\end{lemma}

The real challenge is posed by the task of estimating the locally undetectable dynamics, since it necessitates communicating with neighbors, some of whom might be adversarial. In fact, the traditional metric of graph connectivity which plays a pivotal role in the analysis of fault-tolerant and secure distributed algorithms \cite{sundaramtac,fabio}, cannot capture the requirements to be met by a sensor network for addressing adversarial behavior in the context of distributed state estimation. A simple illustration of this fact is as follows.

\begin{example} Consider a scalar unstable plant monitored by a clique of $N+2$ nodes, as depicted in Figure \ref{fig:example}. Nodes $s_1$ and $s_2$ are the only nodes with non-zero measurements, i.e., they are the source nodes for this system. Although this network is fully connected, the presence of a single adversarial node makes it impossible for \textbf{any} algorithm to guarantee estimation of $x[k]$ at every regular node. This remains true even if every regular node possesses knowledge of the network topology. Specifically, if the adversary compromises one of the two source nodes, then it can behave in a way that makes it impossible for the non-source nodes to distinguish between two different state trajectories of the system, due to the conflicting information from the two source nodes.\footnote{We omit details of such an attack strategy due to space constraints. For centralized systems where $f$ sensors are compromised, \cite{fawzi,joao2} have shown that for recovering the state of the system asymptotically, the system must remain detectable after the removal of any $2f$ sensors.}
\end{example}

The above example alludes to the need for a certain amount of redundancy in the measurement structure \textit{and} the network topology. To this end, in \cite{mitraCDC}, we proposed an algorithm that made use of certain directed acyclic subgraphs in addressing the secure distributed estimation problem; properties of such subgraphs are described below.

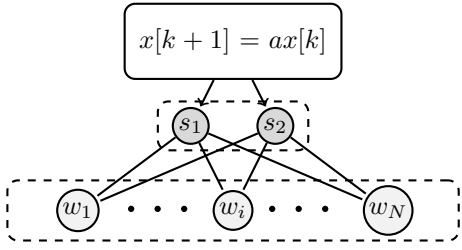
\begin{figure}[t]
\begin{center}
%\hspace*{-1cm}
\begin{tikzpicture}
[->,shorten >=1pt,scale=.75,inner sep=1pt, minimum size=12pt, auto=center, node distance=3cm,
  thick, node/.style={circle, draw=black, thick},]
\tikzstyle{block1} = [rectangle, draw, fill=white, 
    text width=8em, text centered, rounded corners, minimum height=1cm, minimum width=1cm];
 \node [block1]  at (0.75,9) (plant) {$x[k+1]= ax[k]$};
\node [circle, draw, fill=gray!30](n1) at (0,7.5)     (1)  {$s_1$};
\node [circle, draw, fill=gray!30](n2) at (1.5,7.5)   (2)  {$s_2$};
\node [circle, draw, fill=gray!10]  at (-2,6)   (3)  {$w_1$};
\node [circle, draw, fill=gray!10] at (3.5,6) (4) {$w_N$};
\node [circle, draw, fill=gray!10]  at (0.75,6) (5) {$w_i$};
\node [black] at (-0.5,6) () {$\Bigcdots$};
\node [black] at (2,6) () {$\Bigcdots$};
\node (rect) at (0.75,7.5) [draw, dashed, rounded corners, minimum width=2cm, minimum height=0.65cm] {};

\node (rect) at (0.75,6) [draw, dashed,  rounded corners, minimum width=6cm, minimum height=0.8cm] {};
\path[every node/.style={font=\sffamily\small}]

     (plant)
         edge [right] node [] {}(1)
         edge [right] node [] {}(2)
         
     (1)
         edge [-] node [] {}(3)
         edge [-] node [] {}(4)
         edge [-] node [] {}(5)
         
     (2)
         edge [-] node [] {}(3)
         edge [-] node [] {}(4)
         edge [-] node [] {}(5);        
\end{tikzpicture}
\end{center}
\caption{A scalar unstable plant is monitored by a clique of $N+2$ nodes, where $s_1$ and $s_2$ are the only source nodes. A single adversary corrupting either of the two sources can render the distributed state estimation problem impossible, irrespective of the choice of algorithm.}
\label{fig:example}
\vspace{-0.5cm}
\end{figure}

\begin{definition}(\textbf{Mode Estimation Directed Acyclic Graph (MEDAG)}) For each eigenvalue $\lambda_j \in \Omega_{U}(\mathbf{A})$, suppose there exists a spanning subgraph $\mathcal{G}_j = (\mathcal{V},\mathcal{E}_j)$ of $\mathcal{G}$ with the following properties. 
\begin{itemize}
\item[(i)] If $i \in \{\mathcal{V}\setminus\mathcal{S}_j\} \cap \mathcal{R}$, then $|\mathcal{N}^{(j)}_i| \geq 2f+1$, where $\mathcal{N}^{(j)}_i=\{l|(l,i) \in \mathcal{E}_j\}$ represents the neighborhood of node $i$ in $\mathcal{G}_j$.
\item[(ii)] There exists a partition of $\mathcal{R}$ into the sets $\{\mathcal{L}^{(j)}_{0},\mathcal{L}^{(j)}_{1}, \cdots , \mathcal{L}^{(j)}_{T_j}\}$, where $\mathcal{L}^{(j)}_{0} = \mathcal{S}_j \cap \mathcal{R}$, and if $i \in \mathcal{L}^{(j)}_m$ (where $1 \leq m \leq T_j)$, then $\mathcal{N}^{(j)}_i \cap \mathcal{R} \subseteq \bigcup^{m-1}_{r=0} \mathcal{L}^{(j)}_r$.
\end{itemize}

Then, we call $\mathcal{G}_j$ a \textit{Mode Estimation Directed Acyclic Graph (MEDAG)} for $\lambda_j \in \Omega_{U}(\mathbf{A})$.
\label{defn:MEDAG}
\end{definition}

We say a regular node $i\in\mathcal{L}^{(j)}_m$ ``belongs to level $m$", where the levels are indicative of the relative distances of the regular nodes from the source set $\mathcal{S}_j$. The first property indicates that every regular node $i \in \mathcal{V}\setminus{\mathcal{S}_j}$ has at least $(2f+1)$ neighbors in the subgraph $\mathcal{G}_j$, while the second property indicates that all its regular neighbors in such a subgraph belong to levels strictly preceding its own level. In essence, the edges of the MEDAG $\mathcal{G}_j$ represent a medium for transmitting information securely from the source nodes $\mathcal{S}_j$ to the non-source nodes, by preventing the adversaries from forming a bottleneck between such nodes. Intuitively, this requires redundant nodes and edges, and such a requirement is met by the first property of the MEDAG.\footnote{In particular, as regards measurement redundancy, note that for each $\lambda_j \in \Omega_{U}(\mathbf{A})$, a MEDAG $\mathcal{G}_j$ contains atleast $(2f+1)$ source nodes that can detect $\lambda_j$.} Our estimation scheme (described later) relies on a special information flow pattern that requires a node $i$ to listen to \textit{only} its neighbors in $\mathcal{N}^{(j)}_i$ for estimating $z^{(j)}[k]$. The second property of a MEDAG then indicates that nodes in level $m$ only use the estimates of regular nodes in levels $0$ to $m-1$ for recovering $z^{(j)}[k]$. The implications of such properties will become apparent during the analysis of our proposed algorithms.

Before proceeding further, we need to understand the properties of the baseline communication graph $\mathcal{G}$ that guarantee the existence of a MEDAG $\mathcal{G}_j, \forall \lambda_j \in \Omega_{U}(\mathbf{A})$. To this end, we require the following definitions and result from \cite{mitraCDC}.

\begin{definition}($r$-\textbf{reachable set}) For a graph $\mathcal{G}=(\mathcal{V,E})$, a set $\mathcal{S} \subset \mathcal{V}$,  and an integer $r \in \mathbb{Z}_{>0}$, $\mathcal{S}$ is an \textit{$r$-reachable set} if there exists an $i \in \mathcal{S}$ such that $|\mathcal{N}_i \setminus \mathcal{S}| \geq r$. 
\end{definition}

\begin{definition}(\textbf{strongly} $r$-\textbf{robust graph} \textit{w.r.t.} $\mathcal{S}_j$) For $r \in \mathbb{Z}_{>0}$ and $\lambda_j \in \Omega_{U}(\mathbf{A})$, a graph $\mathcal{G}=(\mathcal{V,E})$ is \textit{strongly $r$-robust w.r.t. to the set of source nodes $\mathcal{S}_j$}, if for any non-empty subset $\mathcal{C} \subseteq \mathcal{V}\setminus\mathcal{S}_j$, $\mathcal{C}$ is $r$-reachable.
\label{def:strongrobust}
\end{definition}

\begin{lemma} If $\mathcal{G}$ is strongly $(2f+1)$-robust w.r.t. $\mathcal{S}_j$, for some $\lambda_j \in \Omega_{U}(\mathbf{A})$, then $\mathcal{G}$ contains a MEDAG $\mathcal{G}_j$ for $\lambda_j$.
\label{lemma:graph}
\end{lemma}

Given a $\lambda_j \in \Omega_{U}(\mathbf{A})$, there might be more than one subgraph that satisfies the definition of a MEDAG $\mathcal{G}_j$. In \cite{mitraCDC}, we proposed a distributed algorithm that allowed each node $i$ to identify the sets $\mathcal{N}^{(j)}_i, \forall \lambda_j \in \mathcal{UO}_i$, by explicitly constructing a \textit{specific} MEDAG $\mathcal{G}_j$ for each $\lambda_j \in \mathcal{UO}_i$. In this paper, we assume that these MEDAGs have already been constructed during a distributed design phase using such an algorithm, to inform each node $i$ of the set $\mathcal{N}^{(j)}_i, \forall \lambda_j \in \mathcal{UO}_i$. It will be important to keep in mind that the sets $\mathcal{N}^{(j)}_i$ are time-invariant as they correspond to specific MEDAGs.

\section{Secure State Estimation over Time-varying Networks}
\label{sec:SWLFSE}
In this section, we develop an algorithm that enables each regular node to estimate its locally undetectable portion subject to arbitrary adversarial attacks \textit{and} intermittent communication losses that satisfy the following criterion.

\begin{assumption} There exists $T \in \mathbb{Z}_{>0}$ such that $\forall k \geq T, \bigcup^{T}_{\tau=0}\mathcal{G}[k-\tau]$ contains the MEDAG $\mathcal{G}_j$ for each $\lambda_j \in \Omega_{U}(\mathbf{A})$.
\end{assumption}

Note that under the above communication failure model, $\mathcal{G}[k]$ may not contain the specific MEDAGs constructed during the design phase for some (or all) $k$, thereby precluding direct use of the technique developed in \cite{mitraCDC}. However, such MEDAGs will be preserved in the union graph over the interval $[k-T,k], \forall k\geq T$. For this model, we assume that all estimates being transmitted by regular nodes are properly time-stamped, and propose the following algorithm. 

For each $\lambda_j \in \mathcal{UO}_i$, a regular node $i$ updates its estimate of ${z}^{(j)}[k]$ in the following manner. 

\begin{enumerate}
\item At every time-step $k$, node $i$ collects the \textit{most recent} estimate of $z^{(j)}[k]$ received from each node $l \in \mathcal{N}^{(j)}_i$, along with the corresponding time-stamp $\phi_{il}[k] \in \mathbb{Z}_{\geq 0}$. It then evaluates the delay $\tau_{il}[k] =k-\phi_{il}[k]$ and computes the quantity $\bar{z}^{(j)}_{il}[k]  \triangleq {\lambda_j}^{\tau_{il}[k]}\hat{z}^{(j)}_l[k-\tau_{il}[k]]$.\footnote{For notational simplicity, while considering the eigenvalue $\lambda_j$, we drop the superscript `$j$' on the time-stamp $\phi_{il}[k]$ and the delay $\tau_{il}[k]$.} Prior to receiving the first estimate from a node $l \in \mathcal{N}^{(j)}_i$, the value $\bar{z}^{(j)}_{il}[k]$ is maintained at $0$ by node $i$.\footnote{If node $i$ receives an estimate without a time-stamp from some node in $\mathcal{N}^{(j)}_i \cap \mathcal{A}$, it simply assigns a value of $0$ to such an estimate (without loss of generality). Note that based on Assumption 2, node $i$ is guaranteed to receive a time-stamped estimate from every regular node $l$ in $\mathcal{N}^{(j)}_i$ at least once over every interval of the form $[k-T,k], \, \forall k \geq T$, i.e., for each $l \in \mathcal{N}^{(j)}_i \cap \mathcal{R}$, $\bar{z}^{(j)}_{il}[k]$ will necessarily be of the form $ {\lambda_j}^{\tau_{il}[k]}\hat{z}^{(j)}_l[k-\tau_{il}[k]]$, $\, \forall k\geq T$.}

\item The values $\bar{z}^{(j)}_{il}[k]$ are sorted from largest to smallest; subsequently, the largest $f$ and the smallest $f$ of such values are discarded (i.e., $2f$ values are discarded in all) and $\hat{z}^{(j)}_i[k]$ is updated as
\begin{equation}
\resizebox{0.8\hsize}{!}{$
\hat{z}^{(j)}_i[k+1]=\lambda_j\left(\sum_{l\in\mathcal{M}^{(j)}_i[k]}w^{(j)}_{il}[k]\bar{z}^{(j)}_{il}[k]\right),
\label{eq:updaterule}
$}
\end{equation}
where $\mathcal{M}^{(j)}_i[k] \subset \mathcal{N}^{(j)}_i$ represents the set of nodes whose (potentially) delayed estimates are used by node $i$ at time-step $k$ after the removal of the $2f$ aforementioned values. Node $i$ assigns the consensus weight $w^{(j)}_{il}[k]$ to node $l$ at time-step $k$ for estimating the component of the state corresponding to the eigenvalue $\lambda_j$. The weights $w^{(j)}_{il}[k]$ are non-negative and satisfy $\sum_{l\in\mathcal{M}^{(j)}_i[k]}w^{(j)}_{il}[k]=1, \forall \lambda_j \in \mathcal{UO}_i$.
\end{enumerate}

We refer to the above algorithm as the Sliding Window Local-Filtering based Secure Estimation (SW-LFSE) algorithm. We comment on certain features of this algorithm and then proceed to analyze its convergence properties.

\begin{remark} Like the LFSE algorithm in \cite{mitraCDC}, the SW-LFSE algorithm also relies on a two-stage filtering strategy. Specifically, the first stage of filtering corresponds to a regular node $i \in \mathcal{V}\setminus\mathcal{S}_j$ listening to only its neighbors $\mathcal{N}^{(j)}_i \subseteq \mathcal{N}_i$ in the MEDAG $\mathcal{G}_j$.\footnote{This operation ensures a uni-directional flow of information from the source nodes $\mathcal{S}_j$ (some of whom might also be adversarial) to the rest of the network.} The second stage of filtering requires node $i$ to discard certain extreme values received from nodes in $\mathcal{N}^{(j)}_i$.\footnote{Whereas the first stage of filtering is specific to our distributed state estimation approach, the second stage of filtering is similar to the W-MSR algorithm employed in the secure consensus literature \cite{vaidyacons},\cite{rescons}.} A key point of difference between the algorithms is that in the SW-LFSE approach, at each time-step $k$, node $i$ needs to store the most recent (potentially) delayed  estimate received from each neighbor in $\mathcal{N}^{(j)}_i$. Consequently, we require the regular nodes to possess adequate memory.
\end{remark}

\begin{remark}
Our approach does not require the nodes to have a priori knowledge of the value of $T$ in Assumption 2.
\end{remark}

\begin{remark} Our results will continue to hold if in step $2$ of the SW-LFSE algorithm, node $i$ simply uses the median value of $\bar{z}^{(j)}_{il}[k], l \in \mathcal{N}^{(j)}_i$, in the update rule \eqref{eq:updaterule}. Although this can reduce computation, the present approach offers a degree of freedom in choosing the weights $w^{(j)}_{il}[k]$, that can be potentially leveraged to account for issues like noise.
\end{remark}

\begin{remark} As alluded to earlier in the introduction, this communication-loss model offers the adversaries the additional opportunity of sending false information regarding the time-stamps of their estimates.\footnote{In other words, due to false time-stamp information, the quantity $\hat{z}^{(j)}_l[k-\tau_{il}[k]]$ may not represent the true estimate of an adversarial node $l$ at time $(k-\tau_{il}[k])$. Thus, we resort to a slight abuse of notation here.} Nevertheless, as we establish in the next section, our proposed algorithm is immune to such misbehavior.
\label{rem:opp1}
\end{remark}
%\footnote{On a related note, Remark \ref{rem:whymgreater3} discusses certain subtleties that need to be taken care of while accounting for adversarial behavior that seeks to leverage the specific properties of the stochastic packet dropping model considered in Section \ref{sec:analog_nodelay}.}
\section{Analysis of the SW-LFSE Algorithm}
The following is the main result of this section.

\begin{theorem} Given an LTI system \eqref{eqn:plant} and a measurement model \eqref{eqn:Obsmodel}, let the baseline communication graph $\mathcal{G}$ be strongly $(2f+1)$-robust w.r.t. $\mathcal{S}_j, \forall \lambda_j\in\Omega_{U}(\mathbf{A})$. Furthermore, let Assumptions 1 and 2 be met. Then, the proposed SW-LFSE algorithm guarantees the following despite the actions of any set of $f$-local Byzantine adversaries.

\begin{itemize} 
\item (\textbf{Asymptotic Stability}) Each regular node $i\in\mathcal{R}$ can asymptotically estimate the state of the plant, i.e., $\lim_{k\to\infty} ||\hat{\mathbf{x}}_i[k]-\mathbf{x}[k]||=0, \forall i \in \mathcal{R}$, where $\hat{\mathbf{x}}_i[k]$ is the estimate of the state $\mathbf{x}[k]$ maintained by node $i$.
\item (\textbf{Rate of Convergence}) Let $e^{(j)}_i[k]=\hat{z}^{(j)}_i[k]-z^{(j)}[k]$ denote the error in estimation of the component $z^{(j)}[k]$ by a regular node $i \in \mathcal{V}\setminus\mathcal{S}_j$. Then, if node $i$ belongs to level $q$ of the MEDAG $\mathcal{G}_j$, its estimation error $e^{(j)}_i[k]$ satisfies the following inequality $\forall k \geq (T+1)q$:

\begin{equation}
\resizebox{0.8\hsize}{!}{$
|e^{(j)}_i[k]| \leq \beta^{(j)}{\left[(N-(2f+1)){\left(\frac{|\lambda_j|}{\gamma^{(j)}}\right)}^{(T+1)}\right]}^{q}{(\gamma^{(j)})}^k,$}
\label{eq:conv_rate}
\end{equation}
where $\beta^{(j)} > 0$ and $\gamma^{(j)} \in (0,1)$ are certain constants. 
\end{itemize}
\label{theo:main}
\end{theorem}
\begin{proof}
Note that for each regular node $i$, the state vector $\mathbf{z}[k]$ can be partitioned into the components $\mathbf{z}_{\mathcal{O}_i}[k]$ and $\mathbf{z}_{\mathcal{UO}_i}[k]$ that correspond to the detectable and undetectable eigenvalues, respectively, of node $i$. Based on Lemma \ref{lemma:luen}, we know that node $i$ can estimate $\mathbf{z}_{\mathcal{O}_i}[k]$ asymptotically via a locally constructed Luenberger observer. It remains to show that node $i$ can recover $\mathbf{z}_{\mathcal{UO}_i}[k]$, or in other words, for each $\lambda_j \in \mathcal{UO}_{i}$, we need to prove that $\lim_{k\to\infty}|\hat{{z}}^{(j)}_i[k]-{z}^{(j)}[k]|=0$. Equivalently, we show that for each $\lambda_j \in \Omega_{U}(\mathbf{A})$, every regular node $i \in \mathcal{V}\setminus{\mathcal{S}_j}$ can asymptotically recover ${z}^{(j)}[k]$.

Consider an eigenvalue $\lambda_j \in \Omega_{U}(\mathbf{A})$. Since  $\mathcal{E}[k]\subseteq\mathcal{E}$ for all $k$, Assumption $2$ can hold only if the baseline graph $\mathcal{G}$ contains $\mathcal{G}_j$. The latter follows from the conditions of the Theorem and Lemma \ref{lemma:graph}. Next, based on Assumption 2, notice that for all $k \geq T$,  the union of the graphs over the interval $[k-T,k]$ contains the MEDAG $\mathcal{G}_j$. Recall that the sets $\{\mathcal{L}^{(j)}_{0}, \mathcal{L}^{(j)}_{1}, \cdots, \mathcal{L}^{(j)}_{q}, \cdots \mathcal{L}^{(j)}_{T_j}\}$ form a partition of the set of regular nodes $\mathcal{R}$ in such a MEDAG. We prove the desired result by inducting on the level number $q$. For $q=0$, $\mathcal{L}^{(j)}_{0}=\mathcal{S}_j \cap \mathcal{R}$ from definition, and hence all nodes in level $0$ can estimate ${z}^{(j)}[k]$ asymptotically by virtue of Lemma \ref{lemma:luen}. Next, consider a regular node $i$ in $\mathcal{L}^{(j)}_{1}$ and let $e^{(j)}_i[k] \triangleq \hat{z}^{(j)}_i[k]-z^{(j)}[k]$. We first analyze the SW-LFSE update rule \eqref{eq:updaterule}. To this end, at each time-step $k$, let the neighbor set $\mathcal{N}^{(j)}_i$ of node $i$ be partitioned into the sets $\mathcal{U}^{(j)}_i[k], \mathcal{M}^{(j)}_i[k]$ and $\mathcal{J}^{(j)}_i[k]$, where $\mathcal{U}^{(j)}_i[k]$ and $\mathcal{J}^{(j)}_i[k]$ contain $f$ nodes each, with the highest and the lowest values of $\bar{z}^{(j)}_{il}[k]$ respectively, and $\mathcal{M}^{(j)}_i[k]$ contains the remaining nodes in $\mathcal{N}^{(j)}_i$. At any instant $k$, we can either have $\mathcal{M}^{(j)}_i[k] \cap \mathcal{A}=0$ or $\mathcal{M}^{(j)}_i[k] \cap \mathcal{A} \neq 0$. In the former case, all nodes in $\mathcal{M}^{(j)}_i[k]$ belong to $\mathcal{L}^{(j)}_0=\mathcal{S}_j\cap\mathcal{R}$. In the latter case, when node $i$ uses values transmitted by adversarial nodes in its update rule, it follows from the SW-LFSE algorithm, the $f$-locality of the adversary model, and the fact that $|\mathcal{N}^{(j)}_i| \geq (2f+1)$, that for each $l\in\mathcal{M}^{(j)}_i[k]\cap\mathcal{A}$, there exists a node $u\in\mathcal{U}^{(j)}_i[k]$ and a node $v\in\mathcal{J}^{(j)}_i[k]$ such that both $u,v \in \mathcal{L}^{(j)}_0$, and $\bar{z}^{(j)}_{iv}[k]\leq \bar{z}^{(j)}_{il}[k] \leq \bar{z}^{(j)}_{iu}[k]$, i.e., $\bar{z}^{(j)}_{il}[k]$ can be expressed as a convex combination of $\bar{z}^{(j)}_{iu}[k]$ and $\bar{z}^{(j)}_{iv}[k]$.\footnote{Explicit dependence of $u,v$ on the parameters represented by $i,j,l$ and $k$ is not shown to avoid cluttering of the exposition.} Based on the above discussion and \eqref{eq:updaterule}, it follows that for all $k$, $\hat{z}^{(j)}_i[k+1]$ belongs to the convex hull formed by $\lambda_j \bar{z}^{(j)}_{il}[k]$, $l\in \mathcal{L}^{(j)}_0$. Specifically, there exist weights $\bar{w}^{(j)}_{il}[k]$ such that $\sum_{l\in \mathcal{N}^{(j)}_i \cap \mathcal{L}^{(j)}_0}\bar{w}^{(j)}_{il}[k]=1$, and 
\begin{equation}
\resizebox{0.75\hsize}{!}{$
\hat{z}^{(j)}_i[k+1]= \lambda_j\left(\sum_{l\in \mathcal{N}^{(j)}_i \cap \mathcal{L}^{(j)}_0 }\bar{w}^{(j)}_{il}[k]\bar{z}^{(j)}_{il}[k]\right).
\label{eq:simp_dyn1}
$}
\end{equation}
Since $\sum_{l\in \mathcal{N}^{(j)}_i \cap \mathcal{L}^{(j)}_0}\bar{w}^{(j)}_{il}[k]=1$, and $z^{(j)}[k+1]=\lambda_jz^{(j)}[k]$ based on \eqref{eqn:plant_tr}, we obtain 
\begin{equation}
\resizebox{0.95\hsize}{!}{$
z^{(j)}[k+1]=\lambda_j\left(\sum_{l\in \mathcal{N}^{(j)}_i \cap \mathcal{L}^{(j)}_0}\bar{w}^{(j)}_{il}[k]{\lambda_j}^{\tau_{il}[k]}z^{(j)}[k-\tau_{il}[k]]\right).
\label{eq:true_dyn}
$}
\end{equation}
Based on Assumption 2 and step 1 of the SW-LFSE update rule, we have that for all $k \geq T$, $\bar{z}^{(j)}_{il}[k]={\lambda_j}^{\tau_{il}[k]}\hat{z}^{(j)}_{l}[k-\tau_{il}[k]]$, $l\in\mathcal{N}^{(j)}_i \cap \mathcal{L}^{(j)}_0$. Subtracting \eqref{eq:true_dyn} from \eqref{eq:simp_dyn1}, we then obtain the following error dynamics for all $k \geq T$:
\begin{equation}
\resizebox{0.85\hsize}{!}{$
e^{(j)}_i[k+1]= \lambda_j\left(\sum_{l\in \mathcal{N}^{(j)}_i \cap \mathcal{L}^{(j)}_0}\bar{w}^{(j)}_{il}[k]{\lambda_j}^{\tau_{il}[k]}e^{(j)}_l[k-\tau_{il}[k]]\right).
\label{eq:error_dyn1}
$}
\end{equation} 
Noting that the weights $\bar{w}^{(j)}_{il}[k]$ belong to $[0,1]$, the delay terms $\tau^{}_{il}[k]$ are upper bounded by $T$ for $l\in\mathcal{N}^{(j)}_i\cap\mathcal{R}$, $\lambda_j$ satisfies $|\lambda_j| \geq 1$, and using the triangle inequality, we obtain the following based on \eqref{eq:error_dyn1} for all $k \geq T$:
\begin{equation}
\resizebox{0.85\hsize}{!}{$
|e^{(j)}_i[k+1]| \leq |\lambda_j|^{(T+1)}\hspace{0.75mm}\left(\sum_{l\in \mathcal{N}^{(j)}_i \cap \mathcal{L}^{(j)}_0 }\hspace{.5mm}|e^{(j)}_l[k-\tau_{il}[k]]|\right).
\label{eq:err_dyn2}
$}
\end{equation}
For every $l \in \mathcal{L}^{(j)}_0$, since $e^{(j)}_{l}[k]$ converges asymptotically, and hence exponentially (the local Luenberger observer based dynamics are linear)  based on Lemma \ref{lemma:luen}, there exist constants $c^{(j)}_l > 0$ and $\gamma^{(j)}_l \in (0,1)$ such that $|e^{(j)}_l[k]| \leq c^{(j)}_l{(\gamma^{(j)}_l)}^k$. Let $\beta^{(j)} \triangleq \max_{l\in\mathcal{L}^{(j)}_0} c^{(j)}_l$ and ${\gamma}^{(j)} \triangleq \max_{l\in\mathcal{L}^{(j)}_0} {\gamma}^{(j)}_l$. Then, we obtain the following inequality based on \eqref{eq:err_dyn2} for all $k\geq T$:
\begin{equation}
|e^{(j)}_i[k+1]| \leq |\lambda_j|^{(T+1)}|\mathcal{M}^{(j)}_i[k]|\beta^{(j)}{(\gamma^{(j)})}^{(k-T)}.
\label{eq:err_dyn3}
\end{equation}
Finally, note that based on the rules of the SW-LFSE algorithm, at each time-step $k \geq T$, node $i$ discards $2f$ values and does not use its own estimate in the update rule \eqref{eq:updaterule}. Hence, we have $|\mathcal{M}^{(j)}_i[k]| \leq (N-(2f+1))$. Thus, we obtain \eqref{eq:conv_rate} for $q=1$, implying exponential stability of the error dynamics \eqref{eq:error_dyn1} for all nodes in level $1$, since $\gamma^{(j)} \in (0,1)$. 

Suppose asymptotic stability holds for nodes in all levels from $0$ to $q$ (where $ 1 \leq q \leq T_j-1 $). It is easy to see that the result holds for all nodes in $\mathcal{L}^{(j)}_{q+1}$ as well, by noting that (i) a regular node $i \in \mathcal{L}^{(j)}_{q+1}$ has $\mathcal{N}^{(j)}_i \cap \mathcal{R} \subseteq \bigcup^{q}_{r=0}\mathcal{L}^{(j)}_r$, and (ii) any value $\bar{z}^{(j)}_{il}[k]$ used by node $i$ in the update rule \eqref{eq:updaterule} lies in the convex hull formed by $\bar{z}^{(j)}_{iu}[k], u\in \bigcup^{q}_{r=0}\mathcal{L}^{(j)}_r$. Based on the induction hypothesis, asymptotic stability can then be argued using the same reasoning as the $q=1$ case. Verifying \eqref{eq:conv_rate} is a matter of straightforward algebra.
\end{proof}

The above arguments can be used almost identically to analyze the impact of bounded communication delays (potentially random, time-varying) in the presence of adversaries, for static networks.\footnote{Here, by a bounded communication delay we imply that if $(i,j)\in \mathcal{E}$ and $i,j \in \mathcal{R}$, then any estimate transmitted by node $i$ to node $j$ at time-step $k$, is received by node $j$ no later than time-step $k+T$, for some $T\in\mathbb{Z}_{> 0}$.} We formalize this observation below.

\begin{corollary}
Given an LTI system \eqref{eqn:plant} satisfying Assumption 1, and a measurement model \eqref{eqn:Obsmodel}, let $\mathcal{G}[k]=\mathcal{G} \hspace{2mm} \forall k$, where $\mathcal{G}$ is strongly $(2f+1)$-robust w.r.t. $\mathcal{S}_j, \forall \lambda_j\in\Omega_{U}(\mathbf{A})$. Furthermore, let communication delays between any pair of regular nodes in $\mathcal{G}$ be bounded by some $T \in \mathbb{Z}_{>0}$. Then, the proposed SW-LFSE algorithm provides identical guarantees as in Theorem \ref{theo:main}.
\label{corr1}
\end{corollary}

We conclude this section with the following remark.
\begin{remark} For each unstable/ marginally stable eigenvalue of the system, Theorem \ref{theo:main} and Corollary \ref{corr1} explicitly relate the rates of convergence of the different non-source nodes to the the instability of the system mode under consideration, the delay upper bound $T$, and the respective distances of such nodes (captured by the different levels `$q$' of a MEDAG) from the corresponding sources of information. 
\end{remark}

\section{Secure State Estimation over Analog Erasure Channels}
\label{sec:analog}
In this section, we explore the problem of secure distributed state estimation over a network where each communication link is modeled by an analog erasure channel as defined in \cite{elia}. Specifically, the transmission of information across any link $(i,j) \in \mathcal{E}$ is governed by a random process $\xi_{ij}[k]$ that is memoryless, i.e., $\xi_{ij}[k]$ is i.i.d. over time. Furthermore, across space, the packet dropping processes over different links are independent. For any $k$, the random variable $\xi_{ij}[k]$ follows a Bernoulli fading distribution, i.e., $\xi_{ij}[k]=0$ with erasure probability $p$ and $\xi_{ij}[k]=1$ with probability $(1-p)$; the implications of $\xi_{ij}[k]$ assuming the values $0$ and $1$ will be discussed shortly. 

Our \textbf{objective} in this section will be to design a secure distributed state estimation protocol that guarantees mean-square stability of the estimation error dynamics for each regular node, in the following sense.
\begin{definition}(\textbf{Mean-Square Stability (MSS)})
The estimation error dynamics of the regular nodes is said to be mean-square stable if  $\lim_{k\to\infty} E[||\mathbf{e}_i[k]||^2]=0, \forall i \in \mathcal{R}$, where $\mathbf{e}_i[k]=\hat{\mathbf{x}}_i[k]-\mathbf{x}[k]$, and the expectation is taken with respect to the packet dropping processes $\xi_{ij}[k], (i,j) \in \mathcal{E}$.
\label{defn:MSS}
\end{definition}

\subsection{Channels with no delay}
\label{sec:analog_nodelay}
We first consider the case where $\xi_{ij}[k]=1$ implies that any data packet transmitted by node $i$ at time $k$ is received perfectly by node $j$ at time $k$, and when $\xi_{ij}[k]=0$, such a packet is dropped completely. For this model, we propose a simple algorithm described as follows.

For each $\lambda_j \in \mathcal{UO}_i$, a regular node $i$ updates its estimate of ${z}^{(j)}[k]$ in the following manner.

\begin{itemize}
\item At each time-step $k$, if it receives estimates from at least $(2f+1)$ nodes in $\mathcal{N}^{(j)}_i$, it runs the LFSE algorithm, i.e., it removes the largest $f$ and the smallest $f$ estimates $\hat{z}^{(j)}_l[k], l\in \mathcal{N}^{(j)}_i$ and updates $\hat{z}^{(j)}_i[k]$ as
\begin{equation}
\resizebox{0.8\hsize}{!}{$
\hat{z}^{(j)}_i[k+1]=\lambda_j\left(\sum_{l\in\mathcal{M}^{(j)}_i[k]}w^{(j)}_{il}[k] \hat{z}^{(j)}_l[k]\right),
\label{eqn:LFSE}
$}
\end{equation}
where the set $\mathcal{M}^{(j)}_i[k]$ and the weights $w^{(j)}_{il}[k]$ are defined as in the description of the SW-LFSE algorithm in Section \ref{sec:SWLFSE}. Otherwise, it runs open-loop as follows:
\begin{equation}
\hat{z}^{(j)}_i[k+1]=\lambda_j\hat{z}^{(j)}_i[k].
\label{eqn:open_loop}
\end{equation}
\end{itemize}
The above algorithm provides the following guarantees. 
\begin{theorem}
Given an LTI system (1) satisfying Assumption 1, and a measurement model (2), let the baseline communication graph $\mathcal{G}$ be strongly $(mf+1)$-robust w.r.t. $\mathcal{S}_j, \forall \lambda_j\in\Omega_{U}(\mathbf{A})$, where $m\in\mathbb{Z}_{>0}$. For each $(i,j) \in \mathcal{E}$, let $\xi_{ij}[k]$ be a Bernoulli packet dropping process with erasure probability $p$, that is i.i.d. over time and independent of packet dropping processes over other links. Suppose $m \geq 3$ and that the following is true:\footnote{The choice of $m \geq 3$ is justified in Remark \ref{rem:whymgreater3}.}
\begin{equation}
\rho^2\bar{p} < 1,
\label{eqn:conv_crit}
\end{equation}
where $\rho$ is the spectral radius of $\mathbf{A}$, and 
\begin{equation}
\bar{p} \triangleq 1-\sum_{l=(2f+1)}^{(m-1)f+1}\binom{(m-1)f+1}{l}{(1-p)}^{l}p^{(m-1)f+(1-l)}.
\label{eqn:prob}
\end{equation}
Then, the secure distributed state estimation algorithm described by the update rules \eqref{eqn:LFSE} and \eqref{eqn:open_loop} guarantees mean-square stability in the sense of Definition \ref{defn:MSS}, despite the actions of any $f$-local set of Byzantine adversaries
\label{theo:erasure1}
\end{theorem}
\begin{proof}
Note that the packet dropping processes do not affect the estimation of the locally detectable portions of the state, i.e., each regular node $i$ can recover $\mathbf{z}_{\mathcal{O}_i}[k]$ asymptotically based on Lemma \ref{lemma:luen}. Consider $\lambda_j \in \Omega_{U}(\mathbf{A})$. Since $\mathcal{G}$ is strongly $(mf+1)$-robust w.r.t. $\mathcal{S}_j$, a trivial extension of Lemma \ref{lemma:graph} implies that in the MEDAG $\mathcal{G}_j$, $|\mathcal{N}^{(j)}_i| \geq (mf+1), \forall i \in \{\mathcal{V}\setminus{\mathcal{S}_j}\}\cap\mathcal{R}$. We induct on the level numbers $q$ of such a MEDAG $\mathcal{G}_j$ present in the baseline communication graph $\mathcal{G}$. Let $i$ be a node in level $1$. Let $\mathcal{I}_i[k]$ be an indicator random variable\footnote{To avoid cluttering the exposition, we drop the superscript `$j$' on $\mathcal{I}_i[k]$ and certain other terms throughout the proof, since they can be inferred from context.} such that $\mathcal{I}_i[k]=1$ if node $i$ uses the update rule \eqref{eqn:open_loop} and $\mathcal{I}_i[k]=0$ if node $i$ uses the update rule \eqref{eqn:LFSE}. To make the presentation clear, we make the following assumption. Suppose node $i$ receives estimates from more than $(2f+1)$ nodes in $\mathcal{N}^{(j)}_i$ at a certain time-step $k$. Then, after removing $2f$ estimates based on the LFSE algorithm, it listens to \textit{only a single node} picked arbitrarily from $\mathcal{M}^{(j)}_i[k]$, while running \eqref{eqn:LFSE}.\footnote{The result continues to hold for the general update rule \eqref{eqn:LFSE}.} Combining  \eqref{eqn:LFSE} and \eqref{eqn:open_loop}, we obtain
\begin{equation}
\resizebox{0.8\linewidth}{!}{$
\hat{z}^{(j)}_i[k+1]=\lambda_j \left(\mathcal{I}_i[k]\hat{z}^{(j)}_i[k]+(1-\mathcal{I}_i[k])\hat{z}^{(j)}_l[k]\right),
$}
\end{equation}
where $l \in \mathcal{M}^{(j)}_i[k]$. It is easy to see that the error $e^{(j)}_i[k]=\hat{z}^{(j)}_i[k]-z^{(j)}[k]$ follows the dynamics:
\begin{equation}
\resizebox{0.8\linewidth}{!}{$
{e}^{(j)}_i[k+1]=\lambda_j \left(\mathcal{I}_i[k]{e}^{(j)}_i[k]+(1-\mathcal{I}_i[k]){e}^{(j)}_l[k]\right).
\label{eqn:stoch_err}
$}
\end{equation}
Defining $\sigma^{(j)}_i[k] \triangleq E[{({e^{(j)}_i}[k])}^2]$, and using \eqref{eqn:stoch_err}, we obtain:
\begin{align}
\sigma^{(j)}_i[k+1]=&\hspace{0.5mm}\lambda^2_j E[\mathcal{I}^2_i[k]]\sigma^{(j)}_i[k]+\lambda^2_jE[(1-\mathcal{I}_i[k])^2]\sigma^{(j)}_l[k] \nonumber \\ 
\hspace{2mm}&+\hspace{1.5mm}
\underbrace{2\lambda^2_jE[\mathcal{I}_i[k]-\mathcal{I}^2_i[k]]E[({e^{(j)}_i}[k])({e^{(j)}_l}[k])]}_{g[k]}, \nonumber\\
=&\hspace{0.5mm}\lambda^2_jp^{(j)}_i[k]\sigma^{(j)}_i[k]+\lambda^2_j(1-p^{(j)}_i[k])\sigma^{(j)}_l[k], \nonumber\\
\leq&\hspace{0.5mm} (\lambda^2_j\bar{p}) \, \sigma^{(j)}_i[k]+\lambda^2_j\sigma^{(j)}_l[k],
\label{eqn:stoch_err2}
\end{align}
where $l\in \mathcal{M}^{(j)}_i[k]$ and $p^{(j)}_i[k]$ is the probability that $\mathcal{I}_i[k]=1$. We now justify each of the above steps. For arriving at the first equality, we used the fact that $e^{(j)}_i[k]$ is independent of $\mathcal{I}_i[k]$ for any $i\in\mathcal{R}$, based on the update rules \eqref{eqn:LFSE} and \eqref{eqn:open_loop}, and the nature of the packet dropping processes. The fact that $e^{(j)}_l[k]$ (where $l \in \mathcal{M}^{(j)}_i[k]$) is independent of $\mathcal{I}_i[k]$ requires further arguments. In particular, suppose node $l$ is adversarial and has precise knowledge of the number of packets received by node $i$ at time-step $k$, i.e., suppose node $l$ knows $\mathcal{I}_i[k]$. The estimate $\hat{z}^{(j)}_l[k]$ it transmits to node $i$ might then be influenced by the knowledge of $\mathcal{I}_i[k]$. Irrespective of such knowledge, whenever $l\in\mathcal{M}^{(j)}_i[k]$, based on the LFSE update rule \eqref{eqn:LFSE} and the $f$-locality of the adversarial model, it follows from arguments identical to those in Theorem \ref{theo:main} that $e^{(j)}_l[k]$ can be expressed as a convex combination of $e^{(j)}_u[k]$ and  $e^{(j)}_v[k]$, for some $u,v \in \mathcal{L}^{(j)}_0$. Since such nodes are regular, their errors at time $k$ are independent of $\mathcal{I}_i[k]$, and converge to $0$ since $\mathcal{L}^{(j)}_0=\mathcal{S}_j\cap\mathcal{R}$. Thus, for any $l\in\mathcal{M}^{(j)}_i[k]$, $e^{(j)}_l[k]$ and $\mathcal{I}_i[k]$ are independent, and $\lim_{k\to\infty}\sigma^{(j)}_l[k]=0$. Also, since $\mathcal{I}_i[k]$ is an indicator random variable, $E[\mathcal{I}_i[k]]=E[\mathcal{I}^2_i[k]]$. Hence, $g[k]=0$ leading to the second equality in \eqref{eqn:stoch_err2}.

For arriving at the final inequality, we first note that $p^{(j)}_i[k]$ can potentially vary over time and across different nodes  since the adversarial nodes are allowed to behave arbitrarily. In particular, a compromised node may choose not to transmit estimates even if all out-going communication links from such a node are intact. Thus, since it is impossible to exactly compute $p^{(j)}_i[k]$, we instead seek to upper-bound it. To this end, note that the probability that $\mathcal{I}_i[k]=0$, i.e., the probability that node $i$ receives estimates from at least $(2f+1)$ nodes in $\mathcal{N}^{(j)}_i$ at time $k$, is lower bounded by the probability that it receives estimates from at least $(2f+1)$ nodes in $\mathcal{N}^{(j)}_i\cap\mathcal{R}$ at time $k$. The latter probability can be further lower bounded by $(1-\bar{p})$ (where $\bar{p}$ is given by \eqref{eqn:prob}) by noting that $|\mathcal{N}^{(j)}_i\cap\mathcal{R}| \geq ((m-1)f+1)$ based on the $f$-locality of the fault model. In light of the above discussion, we have $p^{(j)}_i[k] \leq \bar{p}$, leading to the last inequality in \eqref{eqn:stoch_err2}.

Finally, equation \eqref{eqn:conv_crit} implies that $\lambda^2_j\bar{p} < 1$, and in turn guarantees that $\lim_{k\to\infty}\sigma^{(j)}_i[k]=0$, based on Input to State Stability (ISS) and the foregoing discussion. 

Suppose $\lim_{k\to\infty}\sigma^{(j)}_i[k]=0$ for all nodes in levels $0$ to $q$. Consider a node $i \in \mathcal{L}^{(j)}_{q+1}[k]$. Its error dynamics evolves based on \eqref{eqn:stoch_err2}, with $g[k]=0$ for reasons discussed above, and $e^{(j)}_l[k]=\alpha^{(j)}_{il}[k]e^{(j)}_u[k]+(1-\alpha^{(j)}_{il}[k])e^{(j)}_v[k]$, for some $\alpha^{(j)}_{il}[k]\in [0,1]$, and some $u,v \in \bigcup_{r=0}^{q}\mathcal{L}^{(j)}_{r}$. The last argument follows from the LFSE update rule \eqref{eqn:LFSE}. Since $\sigma^{(j)}_u[k]$ and $\sigma^{(j)}_v[k]$ converge to $0$ based on the induction hypothesis, the term $E[(e^{(j)}_u[k])(e^{(j)}_v[k])]$ appearing in $\sigma^{(j)}_l[k]$ can be upper-bounded by $\sqrt{\sigma^{(j)}_u[k]\sigma^{(j)}_v[k]}$ by virtue of the Cauchy-Schwartz inequality. This implies $\lim_{k\to\infty}\sigma^{(j)}_l[k]=0, \forall \, l \in \mathcal{M}^{(j)}_i[k]$. The rest of the proof can be completed following similar arguments as the $q=1$ case.
\end{proof}
\begin{figure}
\begin{center}
\includegraphics[scale=0.3]{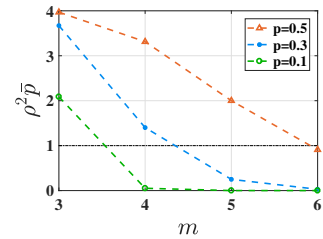}
\caption{Plot illustrating how the effective packet drop probability $\bar{p}$ can be reduced by increasing the level of robustness $m$. For this example, $\rho=2$ and $f=3$.} 
\label{fig:pbar}
\end{center}
\vspace{-0.7cm}
\end{figure}
The term $\bar{p}$ appearing in \eqref{eqn:conv_crit} and \eqref{eqn:prob} can be interpreted as the effective packet drop/erasure probability for the problem under study. With this in mind, the implications of the above result are described as follows.
\begin{remark}(\textbf{Increasing `network robustness' reduces `effective packet drop probability'}) Given knowledge of the spectral radius $\rho$ of $\mathbf{A}$, an upper-bound $f$ on the number of adversaries in the neighborhood of any regular node, and the erasure probability $p$ of the communication network, suppose we are faced with the problem of designing a network topology that guarantees mean-square stability in the sense of Definition \ref{defn:MSS}. Theorem \ref{theo:erasure1} provides an answer to this problem by quantitatively relating our notion of `strong-robustness' in Definition \ref{def:strongrobust} to the effective packet drop probability $\bar{p}$. For instance, as shown in Figure \ref{fig:pbar}, given the parameters $\rho, f$ and $p$, one can generate a plot for $\rho^2\bar{p}$ offline, and choose $m$ to satisfy the MSS criterion $\rho^2\bar{p} < 1$. Subsequently, one can proceed to design a network that is strongly $(mf+1)$-robust w.r.t. $\mathcal{S}_j, \lambda_j \in \Omega_{U}(\mathbf{A})$. It is easy to verify that $\bar{p}$ is monotonically increasing in $p$, and monotonically decreasing in $m$. In other words, for a fixed $\rho$ and $f$, one can tolerate higher erasure probabilities $p$ by increasing the robustness parameter $m$.
\end{remark}

\begin{remark} Note that when $f=0$, i.e., in the absence of adversaries, equation \eqref{eqn:conv_crit} reduces to $\rho^2 p < 1$. This condition is reminiscent of the MSS criterion for remote stabilization of an LTI system over a Bernoulli packet dropping channel \cite{joao}. This observation can be explained by viewing the contribution due to the LFSE update \eqref{eqn:LFSE} (that helps stabilize the error dynamics \eqref{eqn:stoch_err}) as an analogue of the stabilizing input in the remote stabilization problem.
\end{remark}

\begin{remark} We must justify the need for $m\geq3$ in Theorem \ref{theo:erasure1}. For a network that is strongly $(mf+1)$-robust with $m \leq 2$, each adversarial node may follow the simple strategy of never transmitting its estimate. If the adversaries compromise $f$ nodes in some set $\mathcal{N}^{(j)}_i$, where $i\in\mathcal{R}$ and $\lambda_j \in \mathcal{UO}_i$, then such a strategy might cause the regular node $i$ to run open-loop forever based on the algorithm described by the update rules \eqref{eqn:LFSE} and \eqref{eqn:open_loop}. Instead of running open-loop, suppose that if a regular node $i$ does not hear from some neighbor in $\mathcal{N}^{(j)}_i$ at time $k$, it assigns a value of $0$ to the corresponding estimate, and then employs the LFSE update rule \eqref{eqn:LFSE}. Such an approach will in general not work either, due to the following reason. Unlike the communication loss model studied in Section \ref{sec:SWLFSE}, where each regular node was \textbf{guaranteed} to receive estimates from `enough' regular neighbors over bounded intervals of time, no such guarantees can be claimed for the analog erasure channel model studied here. Thus, while strongly $(2f+1)$-robust networks sufficed in Section \ref{sec:SWLFSE}, the choice of $m\geq3$ is in fact necessary in the present context for achieving MSS  based on our specific approach. However, $m=2$ does suffice for certain variants of the analog erasure channel model, as we discuss next.
\label{rem:whymgreater3}
\end{remark}
\vspace{-0.2cm}
\subsection{Channels with erasure and delay}
In this section, we consider a variant of the analog erasure channel that accounts for the presence of random delays. To this end, let $(i,j)\in\mathcal{E}$, and let $\mathbf{v}[k]$ be a message transmitted by node $i$ to node $j$ at time-step $k$. Then, a channel with delay and erasure causes node $j$ to receive the following message:
\begin{equation}
\resizebox{0.7\hsize}{!}{$
\mathbf{r}[k]=\xi_{ij}[k]\mathbf{v}[k]+(1-\xi_{ij}[k])\mathbf{v}[k-\tau_{ij}[k]],
$}
\label{eqn:channel}
\end{equation}  
where $\xi_{ij}[k]$ is the memory-less packet dropping process described earlier, $\tau_{ij}[k] \in \mathbb{Z}_{> 0}$ is a random delay satisfying $1\leq\tau_{ij}[k]\leq T$, and $T \in \mathbb{Z}_{>0}$. In words, the channel output $\mathbf{r}[k]$ is either equal to the present channel input $\mathbf{v}[k]$ with probability $(1-e)$, or equal to a delayed channel input with probability $e$, where the delay is upper bounded by some positive constant $T$. It should be noted that the erasure channel model  considered here is a generalization of the erasure channel with delay in \cite{elia}, where the delays are constant. For this model, we have the following result.
\begin{proposition}
Given an LTI system (1) satisfying Assumption 1, and a measurement model (2), let the baseline communication graph $\mathcal{G}$ be strongly $(2f+1)$-robust w.r.t. $\mathcal{S}_j, \forall \lambda_j\in\Omega_{U}(\mathbf{A})$. Let each communication link of $\mathcal{G}$ be modeled as a channel with delay and erasure as described by equation \eqref{eqn:channel}. Then, the SW-LFSE algorithm provides identical guarantees as in Theorem \ref{theo:main}, with probability $1$.
\label{prop}
\end{proposition}
\vspace{-0.15cm}
\begin{proof}
 The proof follows from the following simple observation: based on the channel model \eqref{eqn:channel}, note that for each $\lambda_j \in \Omega_{U}(\mathbf{A})$, every regular node $i\in\mathcal{V}\setminus\mathcal{S}_j$ is guaranteed to receive a state estimate that is at most $T$ time-steps delayed, from each of its regular neighbors in $\mathcal{N}^{(j)}_i$, at every time-step $k$, $\forall k \geq T$. This corresponds to a special case of the bounded delay model in Corollary \ref{corr1}, and the result thus follows.
\end{proof}
\vspace{-0.4cm}
\section{Conclusion}
\vspace{-0.1cm}
We developed secure distributed state estimation algorithms that account for adversarial nodes in the presence of communication losses, both deterministic and stochastic. For the former case, we characterized the convergence rate of our algorithm in terms of certain system and network properties, and for the latter scenario involving analog erasure channels, we established that our notion of `strong-robustness' plays an important role in tolerating high erasure probabilities while ensuring mean-square stability. As future work, it would be interesting to explore if exploiting sensor memory like the SW-LFSE approach in Section \ref{sec:SWLFSE} can help tolerate higher erasure probabilities for the analog erasure channel model considered in Section \ref{sec:analog_nodelay}.

\bibliographystyle{IEEEtran} 

\bibliography{refs}

\begin{thebibliography}{10}
\providecommand{\url}[1]{#1}
\csname url@samestyle\endcsname
\providecommand{\newblock}{\relax}
\providecommand{\bibinfo}[2]{#2}
\providecommand{\BIBentrySTDinterwordspacing}{\spaceskip=0pt\relax}
\providecommand{\BIBentryALTinterwordstretchfactor}{4}
\providecommand{\BIBentryALTinterwordspacing}{\spaceskip=\fontdimen2\font plus
\BIBentryALTinterwordstretchfactor\fontdimen3\font minus
  \fontdimen4\font\relax}
\providecommand{\BIBforeignlanguage}[2]{{%
\expandafter\ifx\csname l@#1\endcsname\relax
\typeout{** WARNING: IEEEtran.bst: No hyphenation pattern has been}%
\typeout{** loaded for the language `#1'. Using the pattern for}%
\typeout{** the default language instead.}%
\else
\language=\csname l@#1\endcsname
\fi
#2}}
\providecommand{\BIBdecl}{\relax}
\BIBdecl

\bibitem{survey1}
D.~Estrin, R.~Govindan, J.~Heidemann, and S.~Kumar, ``Next century challenges:
  Scalable coordination in sensor networks,'' in \emph{Proceedings of the
  annual ACM/IEEE international conference on Mobile computing and
  networking}.\hskip 1em plus 0.5em minus 0.4em\relax ACM, 1999, pp. 263--270.

\bibitem{martins3}
S.~Park and N.~C. Martins, ``Design of distributed {LTI} observers for state
  omniscience,'' \emph{IEEE Transactions on Automatic Control}, 2016.

\bibitem{allerton}
A.~Mitra and S.~Sundaram, ``An approach for distributed state estimation of
  {LTI} systems,'' in \emph{{P}roceedings of the 2016 Annual Allerton
  Conference on Communication, Control, and Computing}, 2016, pp. 1088--1093.

\bibitem{mitra16arxiv}
------, ``Distributed observers for {LTI} systems,'' \emph{arXiv preprint
  arXiv:1608.01429}, 2016.

\bibitem{wang}
L.~Wang and A.~Morse, ``A distributed observer for a time-invariant linear
  system,'' in \emph{{P}roceedings of the American Control Conference}, 2017,
  pp. 2020--2025.

\bibitem{fawzi}
H.~Fawzi, P.~Tabuada, and S.~Diggavi, ``Secure estimation and control for
  cyber-physical systems under adversarial attacks,'' \emph{IEEE Transactions
  on Automatic Control}, vol.~59, no.~6, pp. 1454--1467, 2014.

\bibitem{pasqualetti}
F.~Pasqualetti, F.~D{\"o}rfler, and F.~Bullo, ``Attack detection and
  identification in cyber-physical systems,'' \emph{IEEE Transactions on
  Automatic Control}, vol.~58, no.~11, pp. 2715--2729, 2013.

\bibitem{joao2}
M.~S. Chong, M.~Wakaiki, and J.~P. Hespanha, ``Observability of linear systems
  under adversarial attacks,'' in \emph{{P}roceedings of the American Control
  Conference}, 2015, pp. 2439--2444.

\bibitem{bai1}
C.-Z. Bai, F.~Pasqualetti, and V.~Gupta, ``Data-injection attacks in stochastic
  control systems: Detectability and performance tradeoffs,''
  \emph{Automatica}, vol.~82, pp. 251--260, 2017.

\bibitem{bai2}
C.-Z. Bai, V.~Gupta, and F.~Pasqualetti, ``On {K}alman filtering with
  compromised sensors: Attack stealthiness and performance bounds,'' \emph{IEEE
  Transactions on Automatic Control}, 2017.

\bibitem{sec1}
I.~Matei, J.~S. Baras, and V.~Srinivasan, ``Trust-based multi-agent filtering
  for increased smart grid security,'' in \emph{{P}roceedings of the
  Mediterranean Conference on Control \& Automation}, 2012, pp. 716--721.

\bibitem{sec3}
U.~Khan and A.~M. Stankovic, ``Secure distributed estimation in cyber-physical
  systems,'' in \emph{Proceedings of the IEEE International Conference on
  Acoustics, Speech and Signal Processing}, 2013, pp. 5209--5213.

\bibitem{deghat}
M.~Deghat, V.~Ugrinovskii, I.~Shames, and C.~Langbort, ``Detection of biasing
  attacks on distributed estimation networks,'' in \emph{{P}roceedings of the
  IEEE Conference on Decision and Control}, 2016, pp. 2134--2139.

\bibitem{mitraCDC}
A.~Mitra and S.~Sundaram, ``Secure distributed observers for a class of linear
  time invariant systems in the presence of byzantine adversaries,'' in
  \emph{{P}roceedings of the IEEE Conference on Decision and Control}, 2016,
  pp. 2709--2714.

\bibitem{saldana17}
D.~Saldana, A.~Prorok, S.~Sundaram, M.~F. Campos, and V.~Kumar, ``Resilient
  consensus for time-varying networks of dynamic agents,'' in
  \emph{{P}roceedings of the American Control Conference}, 2017, pp. 252--258.

\bibitem{dibaji17}
S.~M. Dibaji and H.~Ishii, ``Resilient consensus of second-order agent
  networks: Asynchronous update rules with delays,'' \emph{Automatica},
  vol.~81, pp. 123--132, 2017.

\bibitem{vaidyacons}
N.~H. Vaidya, L.~Tseng, and G.~Liang, ``Iterative approximate byzantine
  consensus in arbitrary directed graphs,'' in \emph{Proceedings of the ACM
  symposium on Principles of distributed computing}, 2012, pp. 365--374.

\bibitem{rescons}
H.~J. LeBlanc, H.~Zhang, X.~Koutsoukos, and S.~Sundaram, ``Resilient asymptotic
  consensus in robust networks,'' \emph{IEEE Journal on Selected Areas in
  Communications}, vol.~31, no.~4, pp. 766--781, 2013.

\bibitem{Sundaramopt}
S.~Sundaram and B.~Gharesifard, ``Consensus-based distributed optimization with
  malicious nodes,'' in \emph{{P}roceedings of the Annual Allerton Conference
  on Communication, Control and Computing}, 2015.

\bibitem{su}
L.~Su and N.~Vaidya, ``Byzantine multi-agent optimization: Part i,''
  \emph{arXiv preprint arXiv:1506.04681}, 2015.

\bibitem{Byz}
D.~Dolev, N.~A. Lynch, S.~S. Pinter, E.~W. Stark, and W.~E. Weihl, ``Reaching
  approximate agreement in the presence of faults,'' \emph{Journal of the ACM
  (JACM)}, vol.~33, no.~3, pp. 499--516, 1986.

\bibitem{sundaramtac}
S.~Sundaram and C.~N. Hadjicostis, ``Distributed function calculation via
  linear iterative strategies in the presence of malicious agents,'' \emph{IEEE
  Transactions on Automatic Control}, vol.~56, no.~7, pp. 1495--1508, 2011.

\bibitem{fabio}
F.~Pasqualetti, A.~Bicchi, and F.~Bullo, ``Consensus computation in unreliable
  networks: A system theoretic approach,'' \emph{IEEE Transactions on Automatic
  Control}, vol.~57, no.~1, pp. 90--104, 2012.

\bibitem{elia}
N.~Elia, ``Remote stabilization over fading channels,'' \emph{Systems \&
  Control Letters}, vol.~54, no.~3, pp. 237--249, 2005.

\bibitem{joao}
J.~P. Hespanha, P.~Naghshtabrizi, and Y.~Xu, ``A survey of recent results in
  networked control systems,'' \emph{Proceedings of the IEEE}, vol.~95, no.~1,
  pp. 138--162, 2007.

\end{thebibliography}

\end{document}